\documentclass[preprint,12pt]{elsarticle}

\usepackage{amssymb}
\usepackage{amsmath}
\usepackage{amsfonts}

\journal{Theoretical Computer Science}

\newtheorem{definition}{Definition}
\newtheorem{proposition}[definition]{Proposition}
\newtheorem{theorem}[definition]{Theorem}
\newtheorem{remark}[definition]{Remark}
\newtheorem{corollary}[definition]{Corollary}
\newtheorem{lemma}[definition]{Lemma}

\newcommand\R{\mathbb{R}}
\newcommand{\Q}{\mathbb{Q}}
\newcommand\RP{\mathbb{R}}

\newcommand\N{\mathbb{N}}
\newcommand\Z{\mathbb{Z}}

\newcommand\mC{\mathcal{C}}
\newcommand\mD{\mathcal{D}}
\newcommand{\len}[1]{\bar{\bar{#1}}}
\newcommand{\GPAC}{\sc gpac}
\newcommand{\norm}[1]{\lVert#1\rVert}

\newenvironment{myequation}{\begin{equation}\begin{aligned}}{\end{aligned}\end{equation}}
\newenvironment{proof}{{\bf Proof}:}{\hfill $\Box$ }

\begin{document}

\begin{frontmatter}



\title{A Framework for Algebraic Characterizations in Recursive Analysis}


\author[LIX]{Olivier Bournez}
\address[LIX]{{\sc LIX, Ecole Polytechnique,} 91128 Palaiseau Cedex, France}
\ead{bournez@lix.polytechnique.fr}

\author[EJUST]{Walid Gomaa \fnref{another_affiliation}}
\address[EJUST]{{\sc Egypt-Japan University of Science and Technology}, Alexandria, Egypt}
\fntext[another_affiliation]{Currently on leave from Alexandria University, Egypt}
\ead{walid.gomaa@ejust.edu.eg}

\author[LORIA,UHP]{Emmanuel Hainry}
\address[LORIA]{{\sc Loria,} BP 239 - 54506 Vand{\oe}uvre-l{\`e}s-Nancy Cedex, France}
\address[UHP]{{\sc Nancy Universit{\'e}, Universit{\'e} Henri Poincar{\'e},} Nancy, France}
\ead{Emmanuel.Hainry@loria.fr}

\begin{abstract}

Algebraic characterizations of the computational aspects
of functions defined over the real numbers provide very effective tool
to understand what computability and complexity over the reals, and generally over
continuous spaces, mean. This is relevant for both communities of computer scientists
and mathematical analysts, particularly the latter who do not understand (and/or like) the language
of machines and string encodings.
Recursive analysis can be considered the most standard framework of computation over continuous
spaces; it is however defined in a very machine specific way which does not leave much to intuitiveness.
Recently several characterizations, in the form of function algebras, of recursively computable functions and some sub-recursive classes were introduced.
These characterizations shed light on the hidden behavior of recursive analysis as they convert complex
computational operations on sequences of real objects to ``simple'' intuitive mathematical operations such as integration or taking limits.
The authors previously presented a framework for obtaining algebraic characterizations at the complexity level over compact domains. The current paper presents a comprehensive extension to that framework. Though
we focus our attention in this paper on functions defined over the whole real line,
the framework, and accordingly the obtained results, can be easily extended
to functions defined over arbitrary domains.

\end{abstract}

\begin{keyword}

recursive analysis \sep polynomial time \sep algebraic characterization
\sep real computation \sep oracle Turing machines \sep complexity class



\end{keyword}

\end{frontmatter}



\section{Introduction}

Many simulations of physical, biological, and/or mathematical phenomena rely on computations over the real numbers.
Hence, it is important to precisely characterize the computability and computational complexity of real processes and operators. The discrete case has been well studied and is now well understood: the Church-Turing thesis
asserts the equivalence of all reasonable models of computations over discrete objects. However, there exist several models of computation over the real numbers which define various notions of computability.
The BSS-model \cite{BluCucShu98,BluShuSma89} named after its discoverers, Blum, Shub and Smale, considers processes of computation using reals in the same way classical models use digits.
While there is no hope to physically realize a BSS-machine, there exist physical machines that compute over the reals.
For example, the Differential Analyzer \cite{Bus31} which was modeled by Shannon as the General Purpose Analog Computer ({\GPAC})\cite{Sha41}. We can also cite algebraically defined classes of computable functions \cite{Moo96}, and recursive analysis. This paper explores the question of complexity in the latter framework: Recursive Analysis.

Recursive analysis is the most accepted model for continuous computation as it does not suffer from physical hopelessness and is as old as the Turing machine. It was indeed introduced by A. Turing \cite{Tur36}. It was also presented by some of the pioneers of computability theory as a natural way of computing over the reals using classical machines\cite{Grz57,Lac55}.
In this model, a real number is represented as a converging sequence of rational numbers, this sequence can be encoded on an infinite tape that is bounded from the left. Hence, a real function $f$ is computable if there exists an algorithm that can produce a sequence representing $f(x)$ given a sequence representing $x$ . The essentials of recursive analysis are presented in section \ref{sec:recana}.
For a complete recent view of the area, the reader may refer to the monograph \cite{Wei00} or the tutorial \cite{BraHerWei08}.

There is no hope to unify all approaches of continuous computations. However, recent work by some authors illustrate that relationships can be established between what is computed using Shannon's {\GPAC} and recursive analysis \cite{JOC2007}. Also it was possible to give characterizations of computable functions, in the context of recursive analysis,
with function algebras {\`a} la Moore \cite{FundamentaInformatica2006}. See survey
\cite{CIEChapter2007} for more insights on the relative equivalences among different continuous models.
However, discussions have mainly focused on the computability level.
Connecting at the complexity level models that are known to be related at the computability level is an even more ambitious goal.
An immediate deep problem is that of defining complexity notions for some of the models such as the {\GPAC}.
The main difficulty of this line of research is that there is no common viewpoint of what continuous time and continuous space mean, as shown by several attempts \cite{Moo96,AMP95,Ruo94,CIEChapter2007}.

To illustrate this point note that no notion of complexity is defined
for the {\GPAC}. BSS and recursive analysis both have such a notion.
In the {\GPAC}, the continuity of time is responsible for this difficulty. In
recursive analysis and BSS, time is discrete, whereas space (in BSS) is
continuous, hence, the notion of time complexity exists in those two models.
However, as expected,
there is no relationship between what is polynomially computable for BSS and what is polynomially computable in recursive analysis.
In recursive analysis, complexity relates the precision needed for the function value
with both the time and the precision needed on the input value.
K-I. Ko \cite{Ko91} defines this notion of complexity as well as
the induced results on what it means to belong to a given complexity class and what the complexity of certain mathematical operators are.
Our aim is to characterize the complexity classes of real functions defined over arbitrary domains
using function algebras in the way of \cite{Moo96}. This line of research has already
produced significant results such as \cite{GCharacterizing,Ijuc} which respectively compares the notion of complexity over rational and real numbers and present a framework for characterizing complexity in recursive analysis over compact domains.

This paper is based on numerous recent developments in the field of computability in recursive analysis and extends or reproves those results. For example, the results from \cite{CMC02} or from \cite{FundamentaInformatica2006} are revisited and reproved as applications of our framework.
What this paper aims to provide can be seen from several viewpoints.
First, it is a step in the direction of implicit complexity for real computation.
In this context, it can be understood as a precursor for research in the line of
\cite{Hof99,Jones00,MM00} for the context of recursive analysis.
Second, it gives an algebraic characterization of what is computable in polynomial time
(as well as higher complexity classes) for real functions.
It is a porting of \cite{Cob65,BelCoo92,Clo98} to the reals.
In fact, as the main contribution is a framework for translating discrete characterizations into continuous ones, we use Bellantoni-and Cook characterization and plug it into the reals in order to capture
continuous polynomial time computability.
It would be immediate to have a similar characterization using Cobham's work and/or all the discrete function algebras presented in \cite{Clo98}.
Third, this paper ventures from computability to complexity. In that sense, this article is
similar to \cite{CamOja06} when applied to complexity.
Finally, this paper has a pedagogical value insofar as the algebraic
characterizations we provide make the domain of complexity over the reals clearer, where
such characterizations do not rely on any kind of machine definitions.
For example, they can give easier ways of proving that a given function is not polynomial time computable.
We believe that this is a very natural and intuitive paradigm that avoids discrete machinery (such as Type-2 Turing machines) when talking about continuous computation.

The research done in the current paper has started a couple of years ago. A preliminary
version of that work was presented at the Third International Workshop on Physics
and Computation and later published in the International Journal of Unconventional Computing \cite{Ijuc}.
The current paper is a comprehensive extension of \cite{Ijuc} as follows: (1) whereas the work in \cite{Ijuc}
is restricted to functions defined over the compact interval $[0,1]$, the current paper
considers real functions defined over any arbitrary domain, though, for simplicity, the proofs
are given only for functions defined over the whole real line and (2) many of the proofs in
\cite{Ijuc} are either eliminated or just outlined, whereas the current paper provides
a detailed complete proof of every result.

The paper is organized as follows. Section 1 is an introduction. Section 2 presents a quick review of the related work
done in the area of algebraically characterizing continuous complexity classes. Section 3 introduces the basic concepts
and results of recursive analysis, the framework for continuous computation that we adopt in this paper.
Section 4 is the core of the paper. It starts with a simple preliminary result in Subsection 4.1, followed, in Subsection 4.2,
by integer characterizations of special classes
of real functions, namely the Lipschitz functions and the locally poly-Lipschitz functions.
Subsection 4.3 avoids the Lipschitz assumptions and generalizes these results to all polynomial time computable real functions. In Section 5 we apply the obtained results to algebraically characterize some computability and complexity
classes of real functions. We first, Subsections 5.1 and 5.2, obtain some restatements and extensions of already known results using our framework. In Subsection 5.3 we  provide new results that give algebraic machine independent characterizations of polynomial time computable analysis functions.

\section{Related Work}

We prove our results by relating the notion of polytime (polynomial time) computable functions over the reals
(arbitrary domains) to the corresponding notion over the integers.
This gives a direct way to lift algebraic characterizations of integer computability and complexity to algebraic characterizations of the corresponding analog notions.
Our setting is actually proved to be robust to approximations. One does not need to be able to compute exactly the
computability and or complexity class over the integers, but only some defined approximation of it in order to be able to compute the corresponding class over the reals. This can be seen as a way to reformulate/reprove/reread very nicely some constructions already used in \cite{JOC2007,FundamentaInformatica2006}.

Hence, our framework gives a way to rely on algebraic machine-independent characterizations of computable functions over the integers. Several such characterizations are known \cite{Clo98}.
In particular, Kleene's functions are well known to capture exactly the integer functions computable by
Type-1 Turing machines. Cobham \cite{Cob65}, and later Bellantoni and Cook \cite{BelCoo92}, were among the first to propose algebraically defined classes of polytime computable integer functions. Our main theorem relies on Bellantoni and Cook's ideas in \cite{BelCoo92}.
Other machine independent characterizations of classical discrete computability and complexity classes (see survey \cite{Clo98}) could also be considered.

These results yield a direct understanding of how algebraic characterizations of computability and complexity classes over the integers can be lifted to algebraic characterizations of the corresponding classes over the reals.
Indeed when talking about computability and complexity over the integers, several machine-independent characterizations of computable functions are known \cite{Clo98}. In particular, Kleene's functions are well known to capture exactly the discrete functions computable by Turing machines.
Cobham \cite{Cob65}, and later Bellantoni and Cook \cite{BelCoo92}, were among the first to propose algebraically defined characterizations of polynomial time computable integer functions.
See the survey \cite{Clo98} for some other machine independent characterizations of classical computability and complexity classes.

Notice that our framework is different from the one proposed by Campagnolo and Ojakian in \cite{CamOja06};
in particular, it has the main advantage
of allowing to talk not only about the computability level but also about the complexity level.
As recursive analysis relies on Type-2 Turing machines,
it is natural to wonder whether the complexity results we obtain are in any way related to works such as
\cite{kapron1996,constable1973}, which characterize the Basic Feasible Functions BFFs which are analogous to polytime computable functions.
First, note that BFFs do not exactly correspond to polytime computable functions in the context of recursive analysis. They are in fact incomparable even if we restrain BFFs to functions that make sense
in the context of recursive analysis. Furthermore, the objects in question are not the same; only very specific sequences have sense in the context of recursive analysis and restricting Type-2 functions to those that represent real functions is far from easy. However, our work could be related to \cite{Isaac10} which seeks to characterize both BFFs and polytime computable functions in the sense of recursive analysis using polynomial interpretations.

\section{Essentials of Recursive Analysis}

\label{sec:recana}

In this section, we recall some basic definitions from recursive analysis: see \cite{Wei00,Ko91} for a full detailed presentation.
Let $\N,\Z,\Q,\R$ denote the set of natural numbers, integer numbers, rational numbers, and real numbers respectively.
Let $\mathbb{D}$ denote the set of dyadic numbers, that is, $\mathbb{D}=\{r\in\mathbb{Q}\colon r=\frac{a}{2^b},a\in \Z,b\in\N\}$.
These are the rationals with finite binary representation. For any real number $x$ let $|x|$ denote the absolute value of $x$.

\begin{definition}[Representation of real numbers]
Assume $x\in\R$. A Cauchy sequence representing $x$ is a function
$\varphi_x\colon\mathbb{N}\rightarrow\mathbb{D}$ that converges at a binary rate:
$\forall n\in\mathbb{N}\colon|x-\varphi_x(n)|\le2^{-n}$.
Given $x\in\R$, let $CF_x$ denote the class of Cauchy functions that represent $x$.
\end{definition}

Based on this representation we can now define the computability of functions over the real numbers.

\begin{definition}[Computability of real functions]
\label{dfn:computability_of_real_functions}
Assume a function\\ $f\colon D\subseteq\R\rightarrow\mathbb{R}$.
We say that $f$ is computable if there exists an oracle Turing machine $M^{^{()}}$ such that
for every $x\in D$, for every $\varphi_x\in CF_x$, and for every $n\in\mathbb{N}$ the following holds:

\begin{myequation}
|M^{^{\varphi_x}}(n) - f(x)|\le 2^{-n}
\end{myequation}
\end{definition}

For $x\not\in D$, and for any $\varphi_x\in CF_x$, the behavior of the machine $M^{^{\varphi_x}}$
is undefined. Note that for the rest of this article we only concentrate on two cases:
either $D=[0,1]$ or $D=\R$.

\begin{definition}[Polytime computability of real functions]
If $D=[0,1]$, then we say $f$ is \emph{polytime computable} if there exists a machine $M$ such that the computation time of $M^{^{\varphi_x}}(n)$ is bounded by $p(n)$
for some polynomial $p$. If $D=\R$, then we say $f$ is \emph{polytime computable}
if the computation time of $M^{^{\varphi_x}}(n)$ is bounded by $p(k,n)$
for some polynomial $p$ where $k=\min\{j\colon x\in[-2^j,2^j]\}$.
We will typically call $k$ the extension parameter and $n$ the precision parameter.
\end{definition}

As is evident from the above definitions, in the context of recursive analysis continuity is a necessary condition of computability, though
it is not sufficient. The following definition introduces the notion of `modulus of continuity'
which in some sense quantifies the concept of continuity and provides a useful tool
in the investigation of real continuous computation \cite{GCharacterizing}.

\begin{definition}[Modulus of continuity]
Consider a function $f\colon\R\to\R$. Then $f$ has a modulus of continuity if there exists a function
$m\colon\mathbb{N}^2\to\mathbb{N}$ such that for all $k,n\in\mathbb{N}$ and for all $x,y\in[-2^k,2^k]$ the following holds:

\begin{myequation}
if \; |x-y|\le 2^{-m(k,n)}, \; then \; |f(x)-f(y)|\le 2^{-n}
\end{myequation}

If $f$ is defined over $[0,1]$ the same definition holds except that the parameter $k$ is not necessary
anymore, that is $m\colon\N\to\N$.
\end{definition}

Notice that the existence of modulus of continuity for a function $f$ implies that this function
is continuous. In analogy with \cite[corollary 2.21]{Ko91}, polytime
computability over unbounded domains can be characterized as follows \cite{GCharacterizing}.

\begin{proposition}
\label{thm:characterizing_ptime_over_real_functions_case}
Assume a function $f\colon\R\rightarrow\mathbb{R}$.
Then $f$ is polytime computable iff there exist two integer functions:
$m:\mathbb{N}^2\rightarrow\mathbb{N}$ and $\psi:\mathbb{D}\times\mathbb{N}\rightarrow\mathbb{D}$
such that
\begin{enumerate}
\item{$m$ is a polynomial function and is a modulus for $f$,}
\item{$\psi$ is an approximation function for $f$, that is, for
every $d\in\mathbb{D}$ and $n\in\mathbb{N}$ the following holds:

\begin{myequation}
|\psi(d,n)-f(d)|\le 2^{-n}
\end{myequation}}
\item{$\psi(d,n)$ is computable in time $p(|d|+n)$ for some polynomial $p$.}
\end{enumerate}
\end{proposition}
\begin{proof}
The proof is an extension of the proof of Corollary 2.21 in \cite{Ko91}.
Assume the existence of $m$ and $\psi$ that satisfy the given conditions. Without loss
of generality assume that $m(k,n)=(k+n)^a$ for some $a\in\mathbb{N}$.
Assume an $f$-input $x\in\R$ and let $\varphi\in CF_x$. Assume $n\in\mathbb{N}$.
Let $M^{^{\varphi}}(n)$ be an oracle Turing machine that does the following:
\begin{enumerate}
\item{let $d=\varphi(2)$,}
\item{from $d$ determine the least $k$ such that $x\in[-2^k,2^k]$,}
\item{let $\alpha=m(k,n+1)$,}
\item{let $d'=\varphi(\alpha)$,}
\item{let $e=\psi(d',n+1)$ and output $e$.}
\end{enumerate}

Note that every step of the above procedure can be performed in polynomial time
with respect to both $k$ and $n$. Now
verifying the correctness of $M^{^{\varphi}}(n)$:

\begin{myequation}
|e-f(x)|&\le|e-f(d')|+|f(d')-f(x)|\\
&\le 2^{-(n+1)}+|f(d')-f(x)|,\qquad \mbox{by definition of $\psi$}\\
& \le 2^{-(n+1)}+2^{-(n+1)},\qquad \mbox{$|d'-x|\le 2^{-m_k(n+1)}$ and definition of $m$}\\
& =2^{-n}
\end{myequation}
This completes the first part of the proof.  Now assume $f$ is polytime computable.
In the following we adopt the following notation: for every $x\in\R$ let $\varphi_x^*(n)=\frac{\lfloor 2^nx\rfloor}{2^n}$.
Fix some large enough $k$ and consider any $x\in\R$ such that $len(\lfloor x\rfloor)=k$ ($len(j)$ denotes the length
of the binary representation of the integer $j$), hence
$x\in[-2^k,2^k]$. Since $f$ is polytime computable, there exists an oracle Turing machine $M^{^{()}}$ such that the computation time of
$M^{^{\varphi_x^*}}(n)$ is bounded by $q(k,n)$ for some polynomial $q$.
Fix some large enough $n\in\mathbb{N}$.
Let
\begin{myequation}
n_x=\max\{j\colon\varphi_x^*(j)\textit{ is queried during the computation of }
M^{^{\varphi_x^*}}(n+3)\}
\end{myequation}

Let $d_x=\varphi_x^*(n_x)$. Note that $\varphi_{d_x}^*(j)=\varphi_x^*(j)$ for every $j\le n_x$.
Let $\ell_x=d_x-2^{-n_x}$ and $r_x=d_x+2^{-n_x}$.
Then $\{(\ell_x,r_x)\colon x\in[-2^k,2^k]\}$ is an \emph{open covering} of
the compact interval $[-2^k,2^k]$. By the \emph{Heine-Borel Theorem}, $[-2^k,2^k]$ has a finite covering
$\mathcal{C}=\{(\ell_{x_i},r_{x_i})\colon i=1,\ldots,w\}$. Note that $x_i\in[-2^k,2^k]$ for $i=1,\ldots,w$.
Define $m'\colon\mathbb{N}^2\to\mathbb{N}$ by

\begin{myequation}
m'(k,n)=\max\{n_{x_i}\colon i=1,\ldots,w\}
\end{myequation}

We need to show that $m'$ is a polynomial modulus for $f$.
From the above assumptions,
$n_x\le q(k,n+3)$ which is polynomial in $k,n$. Now assume some $x,y\in[-2^k,2^k]$ such that $x<y$ and $|x-y|\le 2^{-m'_k(n)}$.\\
\noindent\underline{case 1:} $x,y\in(\ell_{x_i},r_{x_i})$ for some $i\in\{1,\ldots,w\}$.
Then $|x-d_{x_i}|< 2^{-n_{x_i}}$ which implies that $\varphi_x^*(j)=\varphi_{d_{x_i}}^*(j)$ for every $j\le n_{x_i}$,
hence $M^{^{\varphi_x^*}}(n+3)=M^{^{\varphi_{d_{x_i}}^*}}(n+3)$.
Now
\begin{myequation}
|f(x)-f(d_{x_i})|&\le|f(x)-M^{^{\varphi_x^*}}(n+3)|+|M^{^{\varphi_x^*}}(n+3)-f(d_{x_i})|\\
&=|f(x)-M^{^{\varphi_x^*}}(n+3)|+|M^{^{\varphi_{d_{x_i}}^*}}(n+3)-f(d_{x_i})|\\
&\le 2^{-(n+3)}+2^{-(n+3)}\\
&=2^{-(n+2)}
\end{myequation}

Similarly, we can deduce that $|f(y)-f(d_{x_i})|\le 2^{-(n+2)}$. Hence, $|f(x)-f(y)|\le|f(x)-f(d_{x_i})|+|f(d_{x_i})-f(y)|
\le 2^{-(n+2)}+2^{-(n+2)}=2^{-(n+1)}$.\\
\noindent\underline{case 2:} There is no $i$ such that $x,y\in(\ell_{x_i},r_{x_i})$. Notice that $\mathcal{C}$
is a covering and by assumption
$|x-y|\le\min\{\frac{1}{2}(r_{x_i}-\ell_{x_i})\colon i=1,\ldots,w\}$. Hence, there must exist $i,j$ such that
$x\in(\ell_{x_i},r_{x_i})$, $y\in(\ell_{x_j},r_{x_j})$, and $\ell_{x_j}<r_{x_i}$. Choose an
arbitrary $z\in(\ell_{x_j},r_{x_i})$.

Then
\begin{align*}
|f(x)-f(y)|&\le |f(x)-f(z)|+|f(z)-f(y)|\\
& \le 2^{-(n+1)}+|f(z)-f(y)|,\quad\textit{applying case 1 to $x,z\in(\ell_{x_i},r_{x_i})$}\\
&\le 2^{-(n+1)}+2^{-(n+1)},\quad\textit{applying case 1 to $y,z\in(\ell_{x_j},r_{x_j})$}\\
&=2^{-n}
\end{align*}

Hence, $m'$ is a polynomial modulus function for $f$. The approximation function can be defined as follows:
for $d\in\mathbb{D}$ and $n\in\N$, let $\psi(d,n)=M^{^{\varphi_d^*}}(n)$.
This completes the proof of the theorem.
\end{proof}

\section{Characterizing Polytime Real Complexity}

In this section, we prove that it is possible to relate polytime computability over the reals to polytime computability over the integers.
We do that in two steps. In the first step, we consider the special case of Lipschitz functions.
In the second step, we discuss how to avoid the Lipschitz hypothesis, and consider general functions. 
Let's first provide a preliminary result to help explaining what we would like to get.

\subsection{A preliminary first result}

A function over the real line can be characterized by the integer projection
of a function over the real plane.
The extra dimension can be viewed as representing the precision of the computed approximation.

\begin{proposition}[Complexity over $\R$ vs. Complexity over $\R^2$]
\label{prop:char_comp_over_R_2_comp}
The following are equivalent:
\begin{enumerate}
\item{a function $f:\R \to \mathbb{R}$ is polytime computable,}
\item{there exists a polytime computable function $g:\R \times \RP \to \R$ such that the following holds
\begin{equation}
\label{eq:4}
\forall x\in\R,\;\forall y\in\N\colon|g(x,y)-yf(x)|\le 1
\end{equation}}
\end{enumerate}
\end{proposition}
\begin{proof}
$(1)\Rightarrow(2):$ is obtained directly by letting $g(x,y)=yf(x)$. From the hypothesis $f$ is
polytime computable and multiplication can be done in polynomial time, hence, $g(x,y)$
is polytime computable. Clearly, Equation \eqref{eq:4} holds.\\
\noindent$(2)\Rightarrow(1):$ Since $g$ is polytime computable, there exists an oracle
machine $N^{^{()}}$ that computes $g$ in polynomial time.
Assume an input $x\in\R$ and a Cauchy sequence $\varphi_x\in CF_x$.
Assume $n\in\mathbb{N}$ and consider an oracle machine $M^{^{\varphi_x}}(n)$ which does the following:
\begin{enumerate}
\item{Simulate the computation of $N^{^{\varphi_x,\varphi_y}}(0)$, (for $\varphi_y(i)=2^{n+1}$)
\begin{enumerate}
\item{whenever $N^{^{()}}$ queries $\varphi_x(i)$, $M^{^{()}}$ queries its own oracle with the same argument
$i$ and returns
$d=\varphi_x(i)$,}
\item{whenever $N^{^{()}}$ queries $\varphi_y(j)$, $M^{^{()}}$ returns $2^{n+1}$,}
\end{enumerate}}
\item{Repeat the last step as long as $N^{^{()}}$ keeps querying,}
\item{Let $e$ be the output of the simulation of $N^{^{()}}$,}
\item{Output $2^{-(n+1)}e$.}
\end{enumerate}

From this procedure we have

\begin{myequation}
\label{eqn:prop:different_way_of_characterizing_computability_over_R_eqn_1}
&|e-g(x,2^{n+1})|\le 1\\
&|2^{-(n+1)}e-2^{-(n+1)}g(x,2^{n+1})|\le 2^{-(n+1)}
\end{myequation}

From the proposition hypothesis:

\begin{myequation}
\label{eqn:prop:different_way_of_characterizing_computability_over_R_eqn_2}
&|g(x,2^{n+1})-2^{n+1}f(x)|\le 1\\
&|2^{-(n+1)}g(x,2^{n+1})-f(x)|\le 2^{-(n+1)}
\end{myequation}

Then
\begin{align*}
|M^{^{\varphi_x}}(n)-&f(x)|\le\\
&|M^{^{\varphi_x}}(n)-2^{-(n+1)}g(x,2^{n+1})|+|2^{-(n+1)}g(x,2^{n+1})-f(x)|\\
&=|2^{-(n+1)}e-2^{-(n+1)}g(x,2^{n+1})|+|2^{-(n+1)}g(x,2^{n+1})-f(x)|\\
&\le 2^{-(n+1)}+|2^{-(n+1)}g(x,2^{n+1})-f(x)|\qquad\textit{from Eq.
\eqref{eqn:prop:different_way_of_characterizing_computability_over_R_eqn_1}}\\
&\le2^{-(n+1)}+2^{-(n+1)}\qquad\textit{from Eq.
\eqref{eqn:prop:different_way_of_characterizing_computability_over_R_eqn_2}}\\
&\le 2^{-n}
\end{align*}

Hence, the machine $M^{^{\varphi_x}}$ correctly computes the function $f(x)$. Furthermore,
it is clear that $M^{^{\varphi_x}}(n)$ operates in polytime in terms of the precision parameter $n$ and the length of $\lfloor x\rfloor$.
\end{proof}

We would like to talk about functions $g$ with assertions like the above but quantification is only done over the integers.
That is to say about assertions like \eqref{eq:4} but with something like $\forall x \in \Z$ instead of
$\forall x \in \R$. Moving to such full integer characterization we are faced with the problem
of how the notion of continuity of real functions
can be transferred to the domain of integers.

\subsection{Lipschitz functions}

For Lipschitz functions this is facilitated by
the fact that such functions provide us with free information about their continuity properties.

\begin{definition}[Lipschitz functions]
\label{dfn:lipschitz_fun}
A real function $f\colon D\subseteq\R\to\R$ is Lipschitz if there exists a constant $K\ge 0$
such that for all $x_1,x_2\in D$ the following holds:
\begin{myequation}
|f(x_1)-f(x_2)|\le K|x_1-x_2|
\end{myequation}
\end{definition}

Then we have the following characterization of polytime computable Lipschitz functions.

\begin{proposition}[Complexity over $\R$ vs Complexity over $\R^2$]
\label{prop:maversion_comp}
Fix any arbitrary constant $\epsilon\ge0$.
Let $f:\R\to\R$ be a Lipschitz function. Then the following are equivalent:
\begin{enumerate}
\item{$f$ is polytime computable,}
\item{there exists a polytime computable function $g\colon\R\times \R\to\R$ such
that:
\begin{equation} \label{eq:etoileepsilon}
\forall x\in \Z,\forall y\in\N^{\ge 1}\colon|g(x,y)-yf(\frac{x}{y})|\le \epsilon
\end{equation}}
\end{enumerate}
\end{proposition}
\begin{proof}
\noindent$(1)\Rightarrow(2):$ Assume $f$ is polytime computable. Define $g$ as follows:

\begin{myequation}
g(x,y)=
\begin{cases}
0 & y=0 \\
y f(\frac{x}{y}) & y \in\N^{\ge1}\\
piecewise\;linear & otherwise
\end{cases}
\end{myequation}

$g$ is polytime computable since $f$ is polytime computable and arithmetic operations
on the reals are polytime computable. Clearly, $g$ satisfies Eq. \eqref{eq:etoileepsilon}.\\

\noindent $(2)\Rightarrow(1):$ Assume that (2) holds and for simplicity assume $\epsilon=1$.
Let $K$ be a Lipschitz constant for $f$ (see Definition \ref{dfn:lipschitz_fun}).
Let $a\in\N$
such that $K\le2^a$. Hence, for all $x,y\in\R$ the following holds: $|f(x)-f(y)|\le 2^a|x-y|$.
Since $g$ is polytime computable, there exists an oracle machine $N^{^{()}}$ which computes $g$
in polynomial time. Assume an input $x\in\R$ and a Cauchy function $\varphi_x\in CF_x$.
Assume $n\in\mathbb{N}$
and consider an oracle machine $M^{^{\varphi_x}}(n)$ that does the following:
\begin{enumerate}
\item{Let $n'=n+2+a$ and let $d=\varphi_x(n')$,}
\item{Then $|d-x|\le2^{-n'}$, hence, it can be assumed without loss of generality that,
$d=\frac{k_1}{2^{n'}}$ for some $k_1\in\Z$,}
\item{Simulate the operation of $N^{^{\varphi_x,\varphi_y}}(0)$: (for $\varphi_y(j)=2^{n'}$)
\begin{enumerate}
\item{whenever $N^{^{()}}$ queries $\varphi_x(i)$, $M^{^{()}}$ returns $k_1$,}
\item{whenever $N^{^{()}}$ queries $\varphi_y(j)$, $M^{^{()}}$ returns $2^{n'}$,}
\end{enumerate}}
\item{Repeat the last step as long as $N^{^{()}}$ keeps querying,}
\item{Let $e$ be the output of $N^{^{()}}$,}
\item{Output $2^{-n'}e$.}
\end{enumerate}

It can be easy seen that the computation time of $M^{^{\varphi_x}}(n)$ is
bounded by a polynomial in terms of $n$ and $k$, where
$n$ is the precision parameter and $k$ is the least positive integer such that
$x\in [-2^k,2^k]$. Now we want to verify that $M^{^{\varphi_x}}(n)$ computes
$f(x)$ within a precision $2^{-n}$. From the computation of $N^{^{()}}$ we have

\begin{myequation}
\label{eqn:char_over_R_3_1}
&|e-g(k_1,2^{n'})|\le 1\\
&|2^{-n'}e-2^{-n'}g(k_1,2^{n'})|\le 2^{-n'}
\end{myequation}

From Eq. \eqref{eq:etoileepsilon} with $\epsilon=1$

\begin{myequation}
\label{eqn:char_over_R_3_2}
&|g(k_1,2^{n'})-2^{n'}f(\frac{k_1}{2^{n'}})|\le 1\\
&|2^{-n'}g(k_1,2^{n'})-f(\frac{k_1}{2^{n'}})|\le 2^{-n'}
\end{myequation}

From Eq. \eqref{eqn:char_over_R_3_1} and \eqref{eqn:char_over_R_3_2} we have
\begin{myequation}
\label{eqn:char_over_R_3_3}
|2^{-n'}e-f(\frac{k_1}{2^{n'}})|\le 2^{-(n'-1)}
\end{myequation}

From the fact that $f$ is Lipschitz we have
\begin{myequation}
\label{eqn:char_over_R_3_4}
|f(\frac{k_1}{2^{n'}})-f(x)|\le 2^a2^{-n'}=2^{-(n+2)}
\end{myequation}

From Eq. \eqref{eqn:char_over_R_3_3} and \eqref{eqn:char_over_R_3_4}  we have

\begin{myequation}
|2^{-n'}e-f(x)|\le 2^{-(n'-1)}+2^{-(n+2)}\le 2^{-n}
\end{myequation}

This completes the proof that $f$ is polytime computable.

\end{proof}

The previous proposition can be generalized to locally Lipschitz functions as follows.

\begin{definition}[Locally Lipschitz functions]
Assume a function $f\colon\R^k\to\R$.
\begin{enumerate}
\item{We say that $f$ is locally Lipschitz if $f$ is Lipschitz on every compact subset of its domain. That is,
for
every compact set $C\subseteq\R^k$ there exists a constant $K_C$ such that for every $\vec{x},\vec{y}\in C$ the following holds:
\begin{myequation}
|f(\vec{x})-f(\vec{y})| \le K_c \norm{\vec{x}-\vec{y}}
\end{myequation}
$\norm{\cdot}$ is any norm, for example, the Euclidean norm.}
\item{Let $C_i = \bar{B}(\vec{0},2^i)$, where $\bar{B}(\vec{0},2^i)$ is the closed ball centered at the origin
$\vec{0}$
with radius $2^i$.
We say that $f$ is locally poly-Lipschitz if $f$ is locally Lipschitz and there exists a sequence of Lipschitz
constants $\{K_{C_i}\in\N\}_{i\in\N}$ that is polytime computable; that is, there exists a Turing machine $M(i)$ that uniformly computes $K_{C_i}$
in time $p(i)$, where $p$ is a polynomial function.}
\end{enumerate}
\end{definition}

Then we have the following version of Proposition \ref{prop:maversion_comp} for locally poly-Lipschitz functions.

\begin{proposition}[Complexity over $\R$ vs Complexity over $\R^2$]
\label{prop:maversion_comp_2}
Fix any arbitrary constant $\epsilon\ge0$. Let $f:\R\to\R$ be a locally poly-Lipschitz function. Then the following are equivalent:
\begin{enumerate}
\item{$f$ is polytime computable,}
\item{there exists a polytime computable function $g\colon\R\times \R\to\R$ such
that:
\begin{equation}
\label{eq:etoileepsilon_1}
\forall x\in \Z,\forall y\in\N^{\ge 1}\colon|g(x,y)-yf(\frac{x}{y})|\le \epsilon
\end{equation}}
\end{enumerate}
\end{proposition}
\begin{proof}
\noindent$(1)\Rightarrow(2):$ Exactly as the proof of $(1)\Rightarrow(2)$ of Proposition \ref{prop:maversion_comp}.\\
\noindent$(2)\Rightarrow(1):$ Similar to the proof of $(2)\Rightarrow(1)$ of Proposition \ref{prop:maversion_comp}.
Only the operation of $M^{^{\varphi_x}}(n)$ needs to be modified as follows.

\begin{enumerate}
\item{Let $j$ be such that $4x\in[-2^j,2^j]$,}
\item{Compute the Lipschitz constant $K_{C_j}$ of the function $f$ (by assumption on $f$ this can be done in
polynomial time),}
\item{Let $a\in\N$ be such that $K_{C_j}\le 2^a$,}
\item{Let $n'=n+2+a$ and let $d=\varphi_x(n')$,}
\item{Then $|d-x|\le2^{-n'}$, hence, it can be assumed without loss of generality that,
$d=\frac{k_1}{2^{n'}}$ for some $k_1\in\Z$,}
\item{Simulate the operation of $N^{^{\varphi_x,\varphi_y}}(0)$: (for $\varphi_y(j)=2^{n'}$)
\begin{enumerate}
\item{whenever $N^{^{()}}$ queries $\varphi_x(i)$, $M^{^{()}}$ returns $k_1$,}
\item{whenever $N^{^{()}}$ queries $\varphi_y(j)$, $M^{^{()}}$ returns $2^{n'}$,}
\end{enumerate}}
\item{Repeat the last step as long as $N^{^{()}}$ keeps querying,}
\item{Let $e$ be the output of $N^{^{()}}$,}
\item{Output $2^{-n'}e$.}
\end{enumerate}

It can be easily verified that $M^{^{\varphi_x}}(n)$ operates in polynomial time and correctly computes $f(x)$
similar to what was done in Proposition \ref{prop:maversion_comp}.
\end{proof}

In order to interrelate with discrete integer complexity classes we employ the following notion of
\emph{approximation}. 

\begin{definition}[Approximation]
\label{dfn:approximation_1}
Let $\mC$ be a class of functions from $\R^2$ to $\R$. Let $\mD$ be a
class of functions from $\Z^2$ to $\Z$. Assume a function $f\colon\R\to\R$.
\begin{enumerate}
\item{We say that $\mC$ \emph{approximates}
$\mD$ if for any function $g \in \mD$, there exists some function $\tilde{g} \in \mC$  such that for all $x,y\in\Z$ we have
\footnote{Notice that the choice of the constants  $\frac{1}{4}$  and $3$ in this definition
is arbitrary.}
\begin{myequation}
|\tilde{g}(x,y)-g(x,y)| \le 1/4
\end{myequation}}
\item{We say that $f$ is \emph{$\mathcal{C}$-definable} if there exists a function $\tilde{g} \in \mC$
such that the following holds
\begin{myequation}
\label{eq:etoile}
\forall x\in \Z,\forall y\in\N^{\ge 1}\colon|\tilde{g}(x,y)-yf(\frac{x}{y})|\le 3
\end{myequation}}
\end{enumerate}
\end{definition}

We then have the following result.

\begin{theorem}\textbf{\emph{(Complexity over $\R$ vs approximate complexity over $\Z^2$)}}
\label{th:lip}
Consider a class $\mC$ of polytime computable real functions that approximate
the class of polytime computable integer functions.
Assume that $f:\R \to \mathbb{R}$ is Lipschitz. Then
$f$ is polytime computable iff $f$ is $\mathcal{C}$-definable.
\end{theorem}
\begin{proof}
Assume that $f$ is polytime computable. By Proposition \ref{prop:maversion_comp},
there exists a polytime computable function $g$ such that \eqref{eq:etoileepsilon} holds with $\epsilon=\frac{3}{4}$

\begin{myequation}
\label{eqn:th_1_1}
\forall x\in \Z,\forall y\in\N^{\ge 1}\colon|g(x,y)-yf(\frac{x}{y})|\le \frac{3}{4}
\end{myequation}

Since $g$ is polytime computable, there exists an oracle machine $M^{^{()}}$ which
efficiently computes $g$. Consider a function $h:\Z^2\to\Z$ where $h(x,y)$ is defined as follows:
(i) simulate the computation of $M^{^{\varphi_x,\varphi_y}}(1)$ where the exact values of $x$ and $y$
are used to answer the oracle queries, (ii) let $e$ be the output of that simulation, and (iii)
return $\lfloor e\rfloor$ as the value of $h(x,y)$.
By the definition of $h$ we have $|h(x,y)-\lfloor g(x,y)\rfloor|\le 1$.
In addition $h$ is polytime computable, hence from the theorem hypothesis there exists some $\tilde{g} \in \mC$ such that
$$\forall x,y\in\Z:|\tilde{g}(x,y) - h(x,y)| \le 1/4$$

Hence,
\begin{equation}
\forall x,y\in \Z\colon|\tilde{g}(x,y)-\lfloor g(x,y)\rfloor|\le 1+\frac14 = \frac{5}{4}
\end{equation}

We have $|g(x,y)-\lfloor g(x,y) \rfloor|\le 1$, then

\begin{myequation}
\label{eqn:th_1_2}
\forall x,y\in\Z:|\tilde{g}(x,y)-g(x,y)|\le\frac{9}{4}
\end{myequation}

Finally, from Eq. \eqref{eqn:th_1_1} and Eq. \eqref{eqn:th_1_2} we have the desired result

\begin{equation}
\forall x\in \Z,\forall y\in\N^{\ge 1}\colon|\tilde{g}(x,y)-yf(\frac{x}{y})|\le \frac{9}{4}+\frac{3}{4}=3
\end{equation}

The other direction follows from Proposition \ref{prop:maversion_comp} with $\epsilon=3$,
observing that the functions in $\mC$ are polytime computable.
\end{proof}

In the right-to-left direction of the previous theorem,
Eq. \eqref{eq:etoile} implicitly provides a way to efficiently approximate $f$
from $\tilde{g}\upharpoonright\Z^2$.
Computability of $f$ is possible,
in particular at the limit points, from the fact that it is
Lipschitz (hence continuous), and efficiency
is possible by the fact that $\tilde{g}$ is polytime computable.
The left-to-right direction
relates polytime computability of real functions to the corresponding classical discrete notion.

Using Proposition \ref{prop:maversion_comp_2} the previous theorem can be generalized to locally poly-Lipschitz functions as follows.

\begin{theorem}\textbf{\emph{(Complexity over $\R$ vs approximate complexity over $\Z^2$)}}
\label{th:lip_2}
Consider a class $\mC$ of polytime computable real functions that approximate
the class of polytime computable integer functions.
Assume that $f:\R \to \mathbb{R}$ is locally poly-Lipschitz. Then
$f$ is polytime computable iff $f$ is $\mathcal{C}$-definable.
\end{theorem}


\subsection{Avoiding the Lipschitz hypothesis}

\newcommand\mysharp[1]{\#_#1}
\newcommand\mysharpa[2]{\#_#1[#2]}

The major obstacle to avoiding the Lipschitz hypothesis is how to implicitly encode the continuity of $f$ in discrete
computations. This is done in two steps:
(1) encoding the modulus of continuity
which provides information at arbitrarily small rational intervals (however,
it does not tell anything about the irrational limit points)
and (2) bounding the
behavior of the characterizing function $g$ both at compact subintervals of its domain and at the integers.

We need another notion of `approximation' that is a kind of converse to that given in Definition
\ref{dfn:approximation_1}.

\begin{definition}[Polytime computable integer approximation]
A function $g: \R^d \to \R$ is said to have a polytime computable integer approximation
if there exists some polytime computable integer function $h: \Z^d \to \Z$
such that
\begin{myequation}
\forall \bar{x} \in \Z^d\colon |h(\bar{x})- g(\bar{x})|\le 1
\end{myequation}
\end{definition}

A sufficient condition is that the restriction of $g$ to the integers ($g\upharpoonright\Z^2$) is polytime computable.
The choice of the constant $1$ is then due to the fact that this is the best estimated error when trying
to compute the floor of $g(\bar{x})$.
Now we define a special class of functions that will be used to implicitly
describe information about the smoothness of real functions. Their role can be compared to that of the
modulus of continuity.

\begin{definition}[$\mysharp{T}$]
Consider a function $T\colon\N\to \N$ and define
$\mysharp{T}:\R^{\ge 1} \to \R$ by
$\mysharpa{T}{x}=2^{T(\lfloor\log_2 x\rfloor)}$.
When $T$ is a polynomial function with
degree $k$ we write $\mysharp{k}$ to simplify the notation, and we let $\mysharp{k}[x]=2^{\lfloor\log_2 x\rfloor^k}$.
\end{definition}

For $x\in\R$ let $\len{x}$ denote the length of the binary representation of $\lfloor x\rfloor$.
Then the following proposition is the non-Lipschitz version
of Proposition \ref{prop:maversion_comp}.

\begin{proposition}\textbf{\emph{(Complexity over $\R$ vs complexity over $\R^3$)}}
\label{prop:maversiond}
Fix an arbitrary constant $\epsilon\ge0$. Then the following are equivalent:
\begin{enumerate}
\item{a function $f:\R \to \mathbb{R}$ is polytime computable,}
\item{
there exists some function $g\colon\R^3\to\R$ such that
\begin{enumerate}
\item{\label{itemaa} $g$ has a polytime computable integer approximation,}
\item{\label{itema} for some integer $k$,
\begin{myequation}
\label{eq:etoileepsilond}
\forall x\in\R\;\forall y,z\in\R^{\ge1},& z> 4|x| \colon\\
& |g(x\mysharpa{k}{yz},y,z)-yf(x)|\le \epsilon
\end{myequation}}
\item{\label{itemc} \label{itemb} for some constant $M$,
\begin{myequation}
\label{eqn:smoothness_of_g}
\forall x_1,x_2&\in\R\;\forall y,z\in\R^{\ge1}, z>\frac{4|x_1|}{\mysharpa{k}{yz}}\colon \\
& |x_1-x_2|\le 1 \Rightarrow |g(x_1,y,z)-g(x_2,y,z)|\le M
\end{myequation}}
\end{enumerate}}
\end{enumerate}
\end{proposition}
\begin{proof}
$(2)\Rightarrow(1):$
For simplicity, assume $\epsilon=1$. Assume there exists a function $g$ that satisfies the
given conditions.  Assume some $x\in\R$ and $n\in\N$.
Let $y=2^n$ and $z=2^{\len{x}+b}$ for some
arbitrary fixed constant $b\ge4$. Then $y,z\ge 1$ and $z>4|x|$.
From condition \eqref{itema} we have

\begin{myequation}
\label{eqn:char_of_compact_1a}
&|g(2^{(\len{x}+n+b)^k}x,y,z)-yf(x)|\le 1\\
&|y^{-1}g(2^{(\len{x}+n+b)^k}x,y,z)-f(x)|\le y^{-1}
\end{myequation}

Let $h$ be a polytime computable integer function with
$$|h(x,y,z)- g(x,y,z) |\le 1$$ for all $x,y,z \in \Z$.
Such a function exists by condition \eqref{itemaa}. Hence,
\begin{equation}
\label{eqn:char_of_compact_2a}
|g(\lfloor 2^{(\len{x}+n+b)^k}x\rfloor,y,z)-h(\lfloor 2^{(\len{x}+n+b)^k}x\rfloor,y,z)|\le 1
\end{equation}

Note that $\frac{2^{(\len{x}+n+b)^k} 4|x|}{\mysharpa{k}{yz}}
=\frac{2^{(\len{x}+n+b)^k}4 |x|}{2^{\lfloor\log {yz}\rfloor^k}}
=\frac{2^{(\len{x}+n+b)^k}4 |x|}{2^{\lfloor\log {2^{n+\len{x}+b}}\rfloor^k}}=4|x|<z$,
hence condition \eqref{itemc} can be applied to get

\begin{equation}
\label{eqn:char_of_compact_3a}
|g(\lfloor 2^{(\len{x}+n+b)^k} x\rfloor,y,z)-g(2^{(\len{x}+n+b)^k} x,y,z)|\le M
\end{equation}

From Equation \eqref{eqn:char_of_compact_2a} and Equation \eqref{eqn:char_of_compact_3a} we have

\begin{myequation}
\label{eqn:char_of_compact_4a}
&|g( 2^{(\len{x}+n+b)^k} x,y,z)-h(\lfloor 2^{(\len{x}+n+b)^k}x\rfloor,y,z)|\le M+1\\
&|y^{-1}g( 2^{(\len{x}+n+b)^k} x,y,z)-y^{-1}h(\lfloor 2^{(\len{x}+n+b)^k}x\rfloor,y,z)|\le (M+1)y^{-1}
\end{myequation}

From Equation \eqref{eqn:char_of_compact_1a} and Equation \eqref{eqn:char_of_compact_4a}

\begin{myequation}
\label{eqn:char_of_compact_5}
&|f(x)-y^{-1}h(\lfloor 2^{(\len{x}+n+b)^k}x\rfloor,y,z)|\le(M+2)y^{-1}\\
&|f(x)-2^{-n}h(\lfloor 2^{(\len{x}+n+b)^k}x\rfloor,2^n,2^{\len{x}+b})|\le(M+2)2^{-n}
\end{myequation}

Using this last equation we can build a polytime oracle Turing machine that computes $f(x)$ as
follows. Assume some $\varphi\in CF_x$. Consider a
machine $M^{^{\varphi}}(n)$ that does the following:

\begin{enumerate}
\item{let $d'=\varphi(2)$,}
\item{let $\len{d}$ denote the length of the binary representation of $\lfloor d'+1\rfloor$,}
\item{let $d=\varphi((\len{d}+n+4)^k+1)$,}
\item{let $w=h(\lfloor 2^{(\len{d}+n+4)^k}d\rfloor,2^n,2^{\len{d}+4})$,}
\item{output $2^{-n}w$.}
\end{enumerate}

It is clear that the computation time of $M^{^\varphi}(n)$
is bounded by a polynomial in terms of $n$ and $\len{x}$.
So it remains to show the correctness of $M^{^\varphi}(n)$.
There are two cases. Assume first that $\len{x}=\len{d}$. In such a case let $b=4$ in the above equations.
By definition of Cauchy sequences we have:

\begin{myequation}
\label{eqn:main_1}
&|d-x|\le 2^{-((\len{d}+n+4)^k+1)}\\
&|2^{(\len{d}+n+4)^k}d-2^{(\len{d}+n+4)^k}x|\le1/2\\
&|\lfloor2^{(\len{d}+n+4)^k}d\rfloor-\lfloor2^{(\len{d}+n+4)^k}x\rfloor|\le1
\end{myequation}

From Equation \eqref{eqn:char_of_compact_2a}, and the fact that $\len{x}=\len{d}$ and $b=4$, we have

\begin{myequation}
\label{eqn:main_2}
&|g(\lfloor 2^{(\len{d}+n+4)^k}x\rfloor,2^n,2^{\len{d}+4})-h(\lfloor 2^{(\len{d}+n+4)^k}x\rfloor,2^n,2^{\len{d}+4})|\le 1\\
&|g(\lfloor 2^{(\len{d}+n+4)^k}d\rfloor,2^n,2^{\len{d}+4})-h(\lfloor 2^{(\len{d}+n+4)^k}d\rfloor,2^n,2^{\len{d}+4})|\le 1
\end{myequation}

Now we want to apply condition \eqref{itemc} with

\begin{align*}
x_1 = \lfloor2^{(\len{d}+n+4)^k}d\rfloor,\qquad & x_2 = \lfloor2^{(\len{d}+n+4)^k}x\rfloor\\
y = 2^n,\qquad & z = 2^{\len{d}+4}\qquad
\end{align*}

From Eq. \eqref{eqn:main_1} we have $|x_1-x_2|\le1$. And

\begin{myequation}
\frac{4|x_1|}{\mysharpa{k}{yz}} &= \frac{4|\lfloor2^{(\len{d}+n+4)^k}d\rfloor|}{2^{{\lfloor\log_2yz\rfloor}^k}}
= \frac{4|\lfloor2^{(\len{d}+n+4)^k}d\rfloor|}{2^{{\lfloor\log_2 (2^n 2^{\len{d}+4})\rfloor}^k}}\\
& = \frac{4|\lfloor2^{(\len{d}+n+4)^k}d\rfloor|}{2^{{(n+\len{d}+4)}^k}}
\le \frac{4\lfloor2^{(\len{d}+n+4)^k}2^{\len{d}}\rfloor}{2^{{(n+\len{d}+4)}^k}}\\
& = \frac{4\cdot 2^{(\len{d}+n+4)^k}2^{\len{d}}}{2^{{(n+\len{d}+4)}^k}}
= 4\cdot 2^{\len{d}} < 2^{\len{d}+4} = z
\end{myequation}

Accordingly

\begin{equation}
\label{eqn:main_4}
|g(\lfloor2^{(\len{d}+n+4)^k}d\rfloor,2^n,2^{\len{d}+4})-g(\lfloor2^{(\len{d}+n+4)^k}x\rfloor,2^n,2^{\len{d}+4})|\le M
\end{equation}

From Equation \eqref{eqn:main_2} and Equation \eqref{eqn:main_4} we have:

\begin{myequation}
\label{eqn:main_5}
|h(\lfloor 2^{(\len{d}+n+4)^k}d\rfloor,2^n,2^{\len{d}+4})-g(\lfloor2^{(\len{d}+n+4)^k}x\rfloor,2^n,2^{\len{d}+4})|\le M+1
\end{myequation}

From Equation \eqref{eqn:main_2} and Equation \eqref{eqn:main_5} we have

\begin{myequation}
\label{eqn:main_6}
&|h(\lfloor 2^{(\len{d}+n+4)^k}x\rfloor,2^n,2^{\len{d}+4})-h(\lfloor 2^{(\len{d}+n+4)^k}d\rfloor,2^n,2^{\len{d}+4})|\le M+2\\
&|2^{-n}h(\lfloor 2^{(\len{d}+n+4)^k}x\rfloor,2^n,2^{\len{d}+4})-\\
&\qquad\qquad\qquad 2^{-n}h(\lfloor 2^{(\len{d}+n+4)^k}d\rfloor,2^n,2^{\len{d}+4})|\le (M+2)2^{-n}\\
\end{myequation}

From Equation \eqref{eqn:char_of_compact_5} and Equation \eqref{eqn:main_6} and the fact that $b=4$ and $\len{x}=\len{d}$:

\begin{equation}
\label{eqn:avoid_lip_2}
|f(x)-2^{-n}h(\lfloor 2^{(\len{d}+n+4)^k}d\rfloor,2^n,2^{\len{d}+4})|\le (M+2)2^{-(n-1)}
\end{equation}

Note that the above algorithm outputs $2^{-n}h(\lfloor 2^{(\len{d}+n+4)^k}d\rfloor,2^n,2^{\len{d}+4})$. Hence,
the algorithm is correct for the case $\len{x}=\len{d}$.
The only other possibility about the relationship between $\len{x}$ and $\len{d}$ is that
$\len{d}=\len{x}+1$. Such a case is equivalent to letting $b=5$ and following the same line of reasoning as above.
Hence, $M^{^\varphi}(n)$ correctly approximates $f(x)$.\\

\noindent$(1)\Rightarrow(2):$ Assume that $f\colon\R\to\R$ is polytime computable.
Hence $f$ has a polynomial modulus
$m(n',n)=(n'+n)^{k-1}$ for some constant $k\in\N^{\ge2}$, where $n'$ and $n$ are the extension and precision parameters
respectively.
Define $g\colon\R^3\to\R$ as follows.


\begin{myequation}
g(x,y,z)=
\begin{cases}
y f(\frac{x}{\mysharpa{k}{yz}}) & y,z\ge 1\\
y f(\frac{x}{\mysharpa{k}{y}}) & y\ge1,z<1\\
y f(\frac{x}{\mysharpa{k}{z}}) & y<1,z\ge1\\
y f(x) & y,z<1\\
\end{cases}
\end{myequation}

Then for every $x\in\R$, every $y,z\in \R^{\ge 1}$, and every $z > 4|x|$ we have

\begin{align*}
|g(x \mysharpa{k}{yz},y,z)-yf(x)|=|yf(\frac{x \mysharpa{k}{yz}}{\mysharpa{k}{yz}})-yf(x)|=0
\end{align*}

Hence, condition \eqref{itema} is satisfied.
Now assume $x_1,x_2 \in \R$ and $y,z\in\R^{\ge1}$ such that $|x_1-x_2|\le1$ and $z>\frac{4|x_1|}{\mysharpa{k}{yz}}$.
Then

\begin{myequation}
\label{eqn:avoid_lip_1}
|g(x_1,y,z)-g(x_2,y,z)|&=|yf(\frac{x_1}{\mysharpa{k}{yz}})-yf(\frac{x_2}{\mysharpa{k}{yz}})|\\
& = y|f(\frac{x_1}{\mysharpa{k}{yz}})-f(\frac{x_2}{\mysharpa{k}{yz}})|
\end{myequation}

We have

\begin{myequation}
|\frac{x_1}{\mysharpa{k}{yz}}-\frac{x_2}{\mysharpa{k}{yz}}|&=\frac{1}{\mysharpa{k}{yz}}|x_1-x_2|\\
&\le\frac{1}{\mysharpa{k}{yz}}=2^{-\lfloor\log{z}+\log{y}\rfloor^k}
\end{myequation}

Since $z>\frac{4|x_1|}{\mysharpa{k}{yz}}$, then $\log{z}\ge len(\lfloor \frac{|x_1|}{\mysharpa{k}{yz}}\rfloor)$
and $\log{z}\ge len(\lfloor \frac{|x_2|}{\mysharpa{k}{yz}}\rfloor)$. Hence, $\log{z}$ is an upper bound to the extension parameter of the
input to $f$ as given by Equation \eqref{eqn:avoid_lip_1}. In addition $\log{y}$ represents the precision parameter.
Hence, by applying the modulus of continuity of $f$, $m(n',n)=(n'+n)^{k-1}$, we have

\begin{myequation}
|g(x_1,y,z)-g(x_2,y,z)|\le|y2^{-\log{y}}|=1
\end{myequation}

Hence, condition \eqref{itemb} is satisfied with $M=1$.
Assume $i_1,i_2,i_3\in\Z$.
The following procedure computes a function $h\colon\Z^3\to\Z$ such that $|h(i_1,i_2,i_3)- g(i_1,i_2,i_3)|\le 1$:
\begin{enumerate}
\item{if $i_2,i_3\ge 1$, then let $j=i_2 i_3$,}
\item{if $i_2\ge1$ and $i_3 < 1$, then let $j=i_2$,}
\item{if $i_2<1$ and $i_3 \ge 1$, then let $j=i_3$,}
\item{if $i_2,i_3< 1$, then let $j=1$,}
\item{let $l=len(j)-1$,}
\item{Shift right the binary representation of $i_1$ by $l^k$ positions, the result would be a dyadic rational $d$
(this corresponds to dividing $i_1$ by $2^{l^k}$),}
\item{Simulate the computation of $f(d)$ assuming large enough though fixed precision. When simulating the oracle, $d$ is
presented exactly,}
\item{Multiply the output of the previous step by $i_2$. Finally, truncate the result to extract the integer part.}
\end{enumerate}

It is clear that all of these steps can be performed in polynomial time
in terms of $len(|i_1|)$, $len(|i_2|)$, and $len(|i_3|)$. The fixed
precision in step 7 can be calculated from the modulus of $f$ given the fact that the error in the output
of the above procedure should be at most $1$. Hence,
condition \eqref{itemaa} is satisfied. This completes the proof of the proposition.
\end{proof}

\begin{remark}~
\begin{enumerate}
\item{The previous proposition can be generalized to any function $f\colon\R^n\to\R$ by appropriately adjusting
the arities of the functions $g$ and $h$ (as well as slightly modifying Definition \ref{dfn:approximation_1} to include
functions with any arity).}
\item{Given a unary function $f\colon\R\to\R$ that is polytime computable the previous proposition
essentially characterizes $f$ with an integer function $h$ that is polytime computable
(see Equation \eqref{eqn:avoid_lip_2} and the associated algorithm). $h$ takes three arguments:
one argument that can be interpreted
as the extension parameter (the last argument, $2^{\len{d}+4}$, in Equation \eqref{eqn:avoid_lip_2}),
another that can be interpreted as the precision parameter, (the middle argument, $2^n$, in Equation \eqref{eqn:avoid_lip_2}),
and the third argument is the value itself
(the first argument, $\lfloor 2^{(\len{d}+n+4)^k} d\rfloor$, in Equation \eqref{eqn:avoid_lip_2}).}
\end{enumerate}
\end{remark}

We need to consider real functions that are well
behaved relative to their restriction to the integers.
For ease of notation, we will use $[a,b]$ to denote either the interval $[a,b]$ or
the interval $[b,a]$, according to whether or not $a<b$.

\begin{definition}[Peaceful functions]~
\begin{enumerate}
\item{A function $g: \R^3 \to \R$ is said to
be peaceful if
\begin{equation}
\forall x\in\R,\forall y,z\in
\N^{\ge 1} \colon g(x,y,z) \in [g(\lfloor x \rfloor ,y,z) ,g(\lceil x \rceil ,y,z)]
\end{equation}}
\item{We say that a class $\mC$ of real functions peacefully approximates some
class $\mD$ of integer functions, if the subclass of peaceful functions of $\mC$ approximates $\mD$.}
\end{enumerate}
\end{definition}



\begin{definition}
Let $\mC$ be a class of functions from $\R^3$ to $\R$.
Let us consider a function $f:\R \to \mathbb{R}$ and a function $T\colon\N\to \N$.
\begin{enumerate}
\item{We say that $f$ is $T$-$\mathcal{C}$-definable if there exists some peaceful function $g \in \mathcal{C}$
such that
\begin{myequation}
\forall x\in \Z\;\forall y,z\in \N^{\ge 1}, z>\frac{4|x|}{\mysharpa{T}{yz}}  \colon|g(x,y,z)-yf(\frac{x}{\mysharpa{T}{yz}})|\le 2
\end{myequation}}
\item{We say that $f$ is $T$-smooth if there exists some integer $M$ such that
\begin{myequation}
\label{eqn:for_continuity_of_f}
\forall x_1,x_2&\in\R\;\forall y,z\in\R^{\ge1}, z>\frac{4|x_1|}{\mysharpa{T}{yz}}:\\
& |x_1-x_2|\le1 \Rightarrow y |f(\frac{x_1}{\mysharpa{T}{yz}})-f(\frac{x_2}{\mysharpa{T}{yz}})|\le M
\end{myequation}}
\end{enumerate}
\end{definition}

Notice the similarity in the role that $\mysharpa{T}{yz}$ plays in the previous definition
and the role of the modulus of continuity of $f$.
Now we can have the non-Lipschitz version of Theorem \ref{th:lip}.

\begin{theorem}\textbf{\emph{(Complexity over $\R$ vs approximate complexity over $\Z^3$)}}
\label{th:two}
Consider a class $\mC$ of real functions that peacefully approximates
polytime computable integer functions.
And whose functions have polytime computable integer approximations
\footnote{A sufficient condition for that is restrictions to integers of functions from $\mC$ are
polytime computable.}. Then the following are equivalent.
\begin{enumerate}
\item{a function $f\colon\R\to\R$ is polytime computable,}
\item{there exists some positive integer $k$ such that
\begin{enumerate}
\item{\label{itemat} $f$ is $n^k$-$\mathcal{C}$-definable,
}
\item{\label{itemf} $f$ is $n^k$-smooth.}
\end{enumerate}}
\end{enumerate}
\end{theorem}
\begin{proof}
$(1)\Rightarrow(2):$
By Proposition \ref{prop:maversiond} there exists a function $g\colon\R^3\to\R$
such that Equation \eqref{eq:etoileepsilond} holds with $\epsilon=3/4$.
Such equation can be rewritten as follows (through change of variables).

\begin{myequation}
\label{eqn:avoid_lip_5}
\forall x\in\R\;\forall y,z\in\R^{\ge1},& z> \frac{4|x|}{\mysharpa{k}{yz}} \colon\\
& |g(x,y,z)-yf(\frac{x}{\mysharpa{k}{yz}})|\le \frac{3}{4}
\end{myequation}

Also, using Proposition \ref{prop:maversiond},
$g$ has a polytime computable integer approximation $h$, that is,

\begin{myequation}
\label{eqn:avoid_lip_4}
\forall x\in\Z, \forall y,z\in\N^{\ge1}: |g(x,y,z)-h(x,y,z)| \le 1
\end{myequation}

Now, by hypothesis, there exists some peaceful real function $\tilde{h} \in \mC$
that approximates $h$

\begin{myequation}
\label{eqn:avoid_lip_3}
\forall x\in\Z,\forall y,z\in \N^{\ge1}:|\tilde{h}(x,y,z) - h(x,y,z)| \le 1/4
\end{myequation}

Then, from Equation \eqref{eqn:avoid_lip_4} and Equation \eqref{eqn:avoid_lip_3} we have

\begin{myequation}
\label{eqn:avoid_lip_6}
\forall x\in\Z,\forall y,z\in \N^{\ge1}\colon |\tilde{h}(x,y,z)- g(x,y,z)|\le 1+\frac14 = \frac{5}{4}
\end{myequation}

From Equation \eqref{eqn:avoid_lip_5} and Equation \eqref{eqn:avoid_lip_6} we have

\begin{myequation}
\label{eqn:from_R_to_approx_Z3_1}
\forall x\in \Z\;&\forall y,z\in \N^{\ge 1}, z>\frac{4|x|}{\mysharpa{k}{yz}}\colon \\
& |\tilde{h}(x,y,z)-yf(\frac{x}{\mysharpa{k}{yz}})|\le \frac{5}{4}+\frac{3}{4}=2
\end{myequation}

This proves (2a) of the current theorem.
Now assume $x_1,x_2\in\R$ such that $|x_1-x_2|\le1$ and assume
$y,z\in\R   ^{\ge1}$ such that $z>\frac{4|x_1|}{\mysharpa{k}{yz}}$.
Then

\begin{align*}
|yf(\frac{x_1}{\mysharpa{k}{yz}})&-yf(\frac{x_2}{\mysharpa{k}{yz}})|\le |yf(\frac{x_1}{\mysharpa{k}{yz}})-g(x_1,y,z)|+\\
& |g(x_1,y,z)-g(x_2,y,z)|+|g(x_2,y,z)-yf(\frac{x_2}{\mysharpa{k}{yz}})|\\
& =\frac{3}{4} + |g(x_1,y,z)-g(x_2,y,z)| + \frac{3}{4},\qquad \textit{applying Equation \eqref{eqn:avoid_lip_5}}\\
& = \frac{3}{2} + M, \qquad \textit{applying Equation \eqref{eqn:smoothness_of_g}}
\end{align*}
This proves \eqref{itemf} of the current theorem.\\

\noindent $(2)\Rightarrow(1):$ This is proven through the use of part (2) of Proposition
\ref{prop:maversiond} as follows. Since $f$ is $n^k$-$\mathcal{C}$-definable,
there exists a peaceful function $g\in\mathcal{C}$ such that

\begin{myequation}
\label{eqn:avoid_lip_7}
\forall x\in \Z\;\forall y,z\in \N^{\ge 1}, z>\frac{4|x|}{\mysharpa{k}{yz}}  \colon|g(x,y,z)-yf(\frac{x}{\mysharpa{k}{yz}})|\le 2
\end{myequation}

From the hypothesis of this theorem $g$ has a polytime
computable integer approximation. Hence, condition \eqref{itemaa} of Proposition \ref{prop:maversiond} is satisfied.
Assume $x_1,x_2\in\Z$ such that $|x_1-x_2|\le 1$ and assume $y,z\in\N^{\ge1}$ such that $z>\frac{4\max\{|x_1|,|x_2|\}}{\mysharpa{k}{yz}}$.
From the hypothesis of the theorem we have $f$ is $n^k$-smooth, hence there exists some positive integer $M$ such that

\begin{myequation}
\label{eqn:avoid_lip_10}
y |f(\frac{x_1}{\mysharpa{k}{yz}})-f(\frac{x_2}{\mysharpa{k}{yz}})|\le M
\end{myequation}

Applying Equation \eqref{eqn:avoid_lip_7} to $x_1$ and $x_2$ we have

\begin{myequation}
\label{eqn:avoid_lip_8}
&|g(x_1,y,z)-yf(\frac{x_1}{\mysharpa{k}{yz}})|\le 2\\
&|g(x_2,y,z)-yf(\frac{x_2}{\mysharpa{k}{yz}})|\le 2
\end{myequation}

Then
\begin{myequation}
\label{eqn:two_successive_values_of_h}
|g(x_1&,y,z)-g(x_2,y,z)|\le |g(x_1,y,z)-yf(\frac{x_1}{\mysharpa{k}{yz}})| + \\
& y |f(\frac{x_1}{\mysharpa{k}{yz}})-f(\frac{x_2}{\mysharpa{k}{yz}})| + |yf(\frac{x_2}{\mysharpa{k}{yz}})-g(x_2,y,z)|\\
& \le 2 + y |f(\frac{x_1}{\mysharpa{k}{yz}})-f(\frac{x_2}{\mysharpa{k}{yz}})| + 2,\;\textit{from Eq. \eqref{eqn:avoid_lip_8}}\\
& = 4 + y |f(\frac{x_1}{\mysharpa{k}{yz}})-f(\frac{x_2}{\mysharpa{k}{yz}})|\\
& \le 4 + M = M' ,\; \textit{from Eq. \eqref{eqn:avoid_lip_10}}
\end{myequation}

Now assume $u,v\in\R$ such that $|u-v|\le1$. And assume $y,z\in\N^{\ge1}$ such that $z>\frac{4\max\{|u|,|v|\}}{\mysharpa{k}{yz}}$. Then

\begin{myequation}
\label{eqn:avoid_lip_11}
|g&(u,y,z)-g(v,y,z)| \le  |g(u,y,z)-g(\lfloor u\rfloor,y,z)| + \\
&\quad\; |g(\lfloor u\rfloor,y,z) - g(\lfloor v\rfloor,y,z)| + |g(\lfloor v\rfloor,y,z) - g(v,y,z)| \\
& \le |g(\lceil u\rceil,y,z)-g(\lfloor u\rfloor,y,z)| + |g(\lfloor u\rfloor,y,z) - g(\lfloor v\rfloor,y,z)| + \\
& \quad\; |g(\lfloor v\rfloor,y,z) - g(\lceil v\rceil,y,z)|,\;\textit{since g is peaceful}\\
& \le M' + |g(\lfloor u\rfloor,y,z) - g(\lfloor v\rfloor,y,z)| + M',\; \textit{from Eq. \eqref{eqn:two_successive_values_of_h}}\\
& \le M' + M' + M', \; \textit{from Eq. \eqref{eqn:two_successive_values_of_h} and the fact that $|u-v|\le1$}\\
& \le 3M'
\end{myequation}

Hence, condition \eqref{itemc} of Proposition \ref{prop:maversiond} is satisfied
(note that only the integer values of $y$ and $z$ are used throughout the proof of $(2)\Rightarrow(1)$ in Proposition \ref{prop:maversiond}).
Assume $w\in\R$ and $y,z\in\N^{\ge 1}$ such that $z>\frac{4|w|}{\mysharpa{k}{yz}}$. Then

\begin{myequation}
|g&(w,y,z)-yf(\frac{w}{\mysharpa{k}{yz}})|\le |g(w,y,z)-g(\lfloor w\rfloor,y,z)| + \\
&\quad\;|g(\lfloor w\rfloor,y,z)-yf(\frac{\lfloor w\rfloor}{\mysharpa{k}{yz}})| +
|yf(\frac{\lfloor w\rfloor}{\mysharpa{k}{yz}})-yf(\frac{w}{\mysharpa{k}{yz}})|\\
& \le 3M'+ |g(\lfloor w\rfloor,y,z)-yf(\frac{\lfloor w\rfloor}{\mysharpa{k}{yz}})| + \\
& \quad\; |yf(\frac{\lfloor w\rfloor}{\mysharpa{k}{yz}})-yf(\frac{w}{\mysharpa{k}{yz}})|, \; \textit{from Eq. \eqref{eqn:avoid_lip_11}}\\
& \le 3M'+ 2 + |yf(\frac{\lfloor w\rfloor}{\mysharpa{k}{yz}})-yf(\frac{w}{\mysharpa{k}{yz}})|,\;\textit{from Eq.
\eqref{eqn:avoid_lip_7}}\\
& \le 3M'+2+M=M'',\;\textit{$f$ is $n^k$-smooth and using Eq. \eqref{eqn:avoid_lip_10}}
\end{myequation}

Hence, condition \eqref{itema} of Proposition \ref{prop:maversiond} is satisfied.
This completes the proof of the theorem.
\end{proof}

The previous results can be generalized to any complexity class as
indicated by the following corollary.

\begin{corollary}
\label{machinchose}
Let $\mathcal{D}$ be some class of functions from $\N$ to $\N$ above the class of polynomial functions and closed under composition.
Consider a class $\mC$ of real functions that peacefully approximates
integer functions computable in time $\mathcal{D}$.
And whose functions have integer approximations computable in time $\mathcal{D}$
\footnote{A sufficient condition for that is restrictions to integers of functions from $\mC$ are computable in time $\mathcal{D}$.}.
Then the following are equivalent.
\begin{enumerate}
\item{a function $f\colon\R\to\R$ is computable in time $\mathcal{D}$,}
\item{there exists some $T \in \mathcal{D}$ such that
\begin{enumerate}
\item{$f$ is $T$-$\mathcal{C}$-definable,}
\item{$f$ is $T$-smooth.}
\end{enumerate}}
\end{enumerate}
\end{corollary}
\begin{proof}
The proof is similar to that of the previous theorem.
It should be noted that if $f$ is computable in time bounded by $\mathcal{D}$ then it has a modulus in $\mathcal{D}$. This is a direct consequence of
the generalization of Theorem 2.19 in
\cite{Ko91} to functions over the whole real line $\R$.
\end{proof}

\section{Applications}

In this section we apply the above results to algebraically characterize some computability and complexity
classes of real functions. We first obtain some restatements and extensions of already known results, using our framework. We then provide new results, in particular, the main result given by Theorems \ref{th:w1} and \ref{th:w2}
and Corollary \ref{cor:w3} which gives algebraic machine independent characterizations of polytime computable analysis functions.

Note that, we obtain characterizations that are valid for functions
defined on any closed interval (potentially infinite), including the whole real line.
In the case of the polytime computable functions, this is a brand new result. In the case of computable functions, this proves the generality of our framework.

A \emph{function algebra} $\mathcal{F}=[\mathcal{B};\mathcal{O}]$ is the smallest class of functions
containing a set of basic functions $\mathcal{B}$ and their closure under a set of operations $\mathcal{O}$.

\subsection{Elementarily computable functions: class $\mathcal{L}$ }

Let us now consider the class $\mathcal{L}$ defined in
\cite{CMC02}: $$\mathcal{L} = [0, 1, -1, \pi, U,\theta_3; COMP, LI]$$
where $\pi$ is the mathematical constant $\pi = 3.14..$, $U$ is the set of projection functions,
$\theta_3(x)=\max\{0, x^3\}$, $COMP$ is the classical composition
operator, and $LI$ is Linear Integration (solution to a system of linear differential equations).
From the constructions in \cite{CMC02}, we know that this class captures the
elementary integer functions.
In addition the following lemma follows from the constructions in \cite{FundamentaInformatica2006}.

\begin{lemma}
$\mathcal{L}$ is a class of real functions computable in elementary time that peacefully approximates total  elementarily computable integer functions.
\end{lemma}

Again using the above results we can obtain characterizations of the class of elementarily computable analysis
functions.

\begin{proposition}[Variation of \cite{CMC02}] A Lipschitz function $f:\R \to \R$ is computable in elementary time iff it is $\mathcal{L}$-definable.
\end{proposition}

\begin{proposition}[Extension of \cite{CMC02}]
Let $f:\R \to \R$ be some $T$-smooth function, for some elementary function $T\colon\N\to\N$.
Then $f$ is computable in elementary time iff it is $T$-$\mathcal{L}$-definable.
\end{proposition}

As in \cite{CMC02,FundamentaInformatica2006}, we can also characterize in a similar way the functions computable in time
$\mathcal{E}_n$ for $n\geq 3$,  where $\mathcal{E}_n$ represents the
$n$-th level of the Grzegorczyk hierarchy.

\subsection{Recursive functions: class $\mathcal{L}_{\mu}$}

Let us now consider the class $\mathcal{L}_{\mu}$ defined in
\cite{FundamentaInformatica2006}: $$\mathcal{L}_{\mu} = [0, 1, U, \theta_3; COMP, LI, \mathit{UMU}]$$
where a zero-finding operator $\mathit{UMU}$ has been added.
This class is known (see \cite{FundamentaInformatica2006}) to extend the class of total (integer) recursive functions.
From the constructions in this latter paper one can show:

\begin{lemma}
$\mathcal{L}_{\mu}$ is a class of computable functions that peacefully approximate the class of total integer recursive functions.
\end{lemma}

And hence, as a consequence of Theorem \ref{th:lip} and Corollary
\ref{machinchose}, we obtain the following result for functions defined on a product of closed (potentially infinite) intervals $\mathcal{D}$:

\begin{proposition}[Variation of \cite{FundamentaInformatica2006}]
A Lipschitz function $f:\mathcal{D} \to \R$ is computable iff it is $\mathcal{L}_{\mu}$-definable.
\end{proposition}

\begin{proposition}[Extension of \cite{FundamentaInformatica2006}]
Let $f:\mathcal{D} \to \R$ be some $T$-smooth function, for some total recursive function $T:\N\to\N$. Then
$f$ is computable iff it is $T$-$\mathcal{L}_{\mu}$-definable.
\end{proposition}

\subsection{Polynomial Time Computable Functions}

We are now ready to provide our main result: an algebraic characterization of polytime computable functions over the reals.
To do so, we define a class of real functions which are essentially extensions
to $\R$ of the Bellantoni-Cook class \cite{BelCoo92}. The latter class
was developed to exactly capture classical discrete polytime computability in an algebraic machine-independent way.
In the next definition any function $f(x_1,\ldots,x_m;y_1,\ldots,y_n)$ has two types
of arguments (see \cite{BelCoo92}): \emph{normal} arguments which come first followed by \emph{safe} arguments,
using `$;$' for separation. For any $n\in\Z$ we call
$[2n,2n+1]$ an even interval and $[2n+1,2n+2]$ an odd interval.

\begin{definition}
Define the function algebra
$$\mathcal{W}=[0,1,+,-,U,c,parity,p;SComp,SI]$$
\begin{enumerate}
\item{zero-ary functions for the constants $0$ and $1$,}
\item{a binary addition function: $+(;x,y)=x+y$,}
\item{a binary subtraction function: $-(;x,y)=x-y$,}
\item{a set of projection functions $U=\{U_i^j\colon i,j\in\N,i\le j\}$
where: \\$U_i^{m+n}(x_1,\ldots,x_m;x_{m+1},\ldots,x_{m+n})=x_i$,}
\item{a polynomial conditional function $c$ defined by:
\footnote{If $x=1$, the conditional is equal to $y$; if $x=0$, it is equal to $z$. Between $0$ and $1$, it stays between $y$ and $z$.}
\begin{myequation}
c(;x,y,z)=xy+(1-x)z
\end{myequation}}
\item{a continuous parity function:
\begin{myequation}
parity(;x)=\max\{0, \frac{\pi}{2} sin(\pi x)\}
\end{myequation}
Hence, $parity(;x)$ is non-zero if and only if $x$ lies inside an even interval.
Furthermore, for any
$n\in\Z$ the following holds: $\int_{2n}^{2n+1}parity(;x)dx=1$.}
\item{a continuous predecessor function $p$ defined by:
\begin{myequation}
p(;x) =  \int_0^{x-1}parity(;t)dt
\end{myequation}
Notice that when $x$ belongs to an even interval $[2n,2n+1]$
$p(;x)$ acts exactly like $\lfloor\frac{x}{2}\rfloor$.
On an odd interval $[2n+1, 2n+2]$, it grows continuously and monotonically from $n$ to $n+1$.}
\item{a safe composition operator $SComp$: Assume a vector of functions $\bar{g}_1(\bar{x};)\in \mathcal{W}$,
a vector of functions
$\bar{g}_2(\bar{x};\bar{y})\in \mathcal{W}$, and a function $h\in \mathcal{W}$ of arity $dim(\bar{g}_1)+dim(\bar{g}_2)$
(where $dim$ denotes the vector length).  Define new function
\begin{myequation}
f(\bar{x};\bar{y})=h(\bar{g}_1(\bar{x};);\bar{g}_2(\bar{x};\bar{y}))
\end{myequation}
It is clear from the asymmetry in this definition that normal
arguments can be repositioned in safe places whereas the opposite can not happen.}
\item{safe integration operator\footnote{Notice that for simplicity we misuse the basic functions (and $p'$) so that their arguments are now in normal positions
(the alternative is to redefine a new set of basic functions with arguments in normal positions).
}$SI$: Assume functions $g,h_0,h_1\in\mathcal{W}$. Let $p'(;x)=p(;x-1)+1$.
Define a new function as the solution of the following ODE:
\begin{myequation}
f(0,\bar{y};\bar{z})=&g(\bar{y};\bar{z})\\
\partial_x f(x,\bar{y};\bar{z})&=parity(x;)[h_1(p(x;),\bar{y};\bar{z},f(p(x;),\bar{y};\bar{z}))\\
&\qquad\qquad\qquad -f(2p(x;),\bar{y};\bar{z})]\\
&+parity(x-1;)[h_0(p'(x;),\bar{y};\bar{z},f(p'(x;),\bar{y};\bar{z}))\\
&\qquad\qquad\qquad\qquad -f(2p'(x;)-1,\bar{y};\bar{z})]
\end{myequation}
This operator closely matches Bellantoni and Cook's predicative recursion on notations:
if $x$ is an even integer, we apply $h_0$ to its predecessor $p(x;)$,
if $x$ is an odd integer, we apply $h_1$ to $p'(x;) = \lfloor \frac{x}{2} \rfloor$.}
\end{enumerate}
\end{definition}

This class $\mathcal{W}$  is based on the Bellantoni-Cook's constructions and the dichotomy of normal/safe arguments in order to have the following
properties, proved by induction.

\begin{proposition}~
\label{prop:properties_of_class_W}
\begin{enumerate}
\item{Class $\mathcal{W}$ preserves the integers, that is for every $f\in\mathcal{W}$ of arity $n$, $f\upharpoonright\Z^n\colon\Z^n\to\Z$.}
\item{Every function in $\mathcal{W}$ is polytime computable.}
\item{Every polytime computable integer function has a peaceful extension in $\mathcal{W}$.}
\end{enumerate}
\end{proposition}
\begin{proof}
\underline{Part I:}
Proof is by induction on the construction of functions in $\mathcal{W}$.
It is easy to see that the constant functions $0$ and $1$, addition, subtraction,
and projections all preserve $\Z$. Given $n\in\Z$ we have $p(;n)=\lfloor\frac{n}{2}\rfloor$ which is an integer.
Given $i,j,k\in\Z$ it is clear that $c(;i,j,k)=ij+(1-i)k$ is an integer.
The parity function is always $0$ at the integer points.
Hence, all the basic functions preserve the integers. Trivially composition preserves the integers.
Let $g,h_0,h_1\in\mathcal{W}$ be functions that preserve $\Z$ and consider the application of the safe integration
operator to define a new function $f\in\mathcal{W}$.
We use strong induction over the discrete values of the integration variable
to show that $f$ preserves $\N$ (for simplicity we restrict to non-negative integers;
also we neglect the arguments $\bar{y}$ and $\bar{z}$ and drop the `;'
in case all arguments of the function are normal).
The base case $f(0)=g$ holds by assumption on $g$. Let $n\in\N^{\ge1}$
and assume $f(j)\in\N$ for every $j\le 2n$, then

\begin{myequation}
\label{eq:class_W_1}
f(2n+1)&=f(2n)+\int_{2n}^{2n+1}parity(x)[h_1(p(x);f(p(x)))-f(2p(x))]dx\\
&=f(2n)+\int_{2n}^{2n+1}parity(x)[h_1(n;f(n))-f(2n)]dx\\
&=f(2n)+[h_1(n;f(n))-f(2n)]\int_{2n}^{2n+1}parity(x)dx\\
&=f(2n)+[h_1(n;f(n))-f(2n)]\cdot 1\\
&=h_1(n;f(n))
\end{myequation}

which is an integer value by the assumption on $h_1$ and the induction hypothesis on $f$.
Similarly, it can be shown that

\begin{myequation}
\label{eq:class_W_2}
f(2n+2)=h_0(n+1;f(n+1))
\end{myequation}
which is also an integer value.
This completes the proof of the first part of the proposition.\\

\noindent\underline{Part II:} Proof is by induction on the construction of functions in $\mathcal{W}$.
It is easy to see that all the basic functions are polytime computable. Composition preserves polytime computability.
Consider a function $f\in\mathcal{W}$ that is defined by safe integration from
$g,h_0,h_1$ where these latter functions are polytime computable.
At $x=0$ we have $f(0)=g$ which is polytime computable by assumption on $g$.
Assume $x\in[2n,2n+1]$ for some $n\in\N$, then

\begin{myequation}
\label{eq:class_W_3}
f(x)&=g+\int_0^{x}parity(u)[h_1(p(u);f(p(u)))-f(2p(u))]du\\
&=f(2n)+\int_{2n}^{x}parity(u)[h_1(n;f(n))-f(2n)]du\\
&=f(2n)+[h_1(n;f(n))-f(2n)]\int_{2n}^{x}parity(u)du\\
&=f(2n)+[f(2n+1)-f(2n)]\int_{2n}^{x}parity(u)du, \textit{\; from Eq. \eqref{eq:class_W_1}}\\
&=f(2n)+[f(2n+1)-f(2n)](p(x+1) - p(2n+1))\\
&=f(2n)+[f(2n+1)-f(2n)](p(x+1) - n)
\end{myequation}

Similarly, over odd intervals $[2n+1,2n+2]$ we have

\begin{myequation}
\label{eq:class_W_4}
f(x)=f(2n+1)+[f(2n+2)-f(2n+1)](p(x)-n)
\end{myequation}

From Eq. \eqref{eq:class_W_1} and Eq. \eqref{eq:class_W_2} we see that
the safe integration operator exactly simulates the behavior of the safe recursion operator of the
Bellantoni-Cook class.
So given the assumption that $g,h_0,h_1$ are polytime computable, we have $f\upharpoonright\N$
is polytime computable. Furthermore, from the base case $p(x)$ is polytime computable.
This completes the proof of the second part of the proposition.\\

\noindent\underline{Part III:}
First we use induction to show that every function in
the Bellantoni-Cook class (which captures integer polytime computability) has an extension in $\mathcal{W}$.
The Bellantoni-Cook class is defined by: $B= [0,U,s_0,s_1,pr,cond;\\SComp,SRec]$, see \cite{BelCoo92}.
The functions $0,U\in\mathcal{W}$ are extensions of the corresponding functions in $B$.
Define functions $\tilde{s}_i\in\mathcal{W}$ by $\tilde{s}_i(;x)=2x+i$ where $i\in\{0,1\}$.
Then $\tilde{s}_i$ are extensions of the successor functions $s_i$.
The predecessor function is defined as follows: $pr(;n)=\lfloor\frac{n}{2}\rfloor$.
From our definition of the class $\mathcal{W}$ we have $p\upharpoonright\N=pr$,
hence $p$ is an extension of the predecessor function $pr$.
Define a function $c_d\in\mathcal{W}$ as follows.

\begin{myequation}
c_d(;x,y,z)=c(;x-2p(;x),z,y)
\end{myequation}

Assume $x=2n$ for $n\in\N$. Then $c_d(;2n,y,z)=c(;2n-2p(;2n),z,y)=c(;0,z,y)=y$.
Now assume $x=2n+1$. Then $c_d(;2n+1,y,z)=c(;2n+1-2p(;2n+1),z,y)=c(;2n+1-2n,z,y)=c(;1,z,y)=z$.
So $c_d\upharpoonright\N^3=cond$, hence it is an extension of the conditional function $cond$.
The case for safe composition $SComp$ is easy.
Now assume $f\in B$ that is defined by safe recursion from $g,h_0,h_1\in B$.
Assume $\tilde{g},\tilde{h}_0,\tilde{h}_1\in\mathcal{W}$ are extensions
of $g,h_0,h_1$. Define the function $\tilde{f}\in\mathcal{W}$ by safe integration from $\tilde{g},\tilde{h}_0,\tilde{h}_1$. We claim that $\tilde{f}$ is an extension of $f$.
Proof is by strong induction on the recursion/integration variable.
At the base case we have $\tilde{f}(0)=\tilde{g}=g=f(0)$.
Let $n\in\N$, then from the proof of the first part of the proposition
we have:

\begin{myequation}
\tilde{f}(2n+1)&=\tilde{h}_1(n;\tilde{f}(n))=\tilde{h}_1(n;f(n)),\;\textit{from induction over n}\\
&=h_1(n;f(n)),\;\textit{by assumption on } \tilde{h}_1\\
&=f(2n+1),\;\textit{by definition of safe recursion}
\end{myequation}

Similarly, it can be shown that $\tilde{f}(2n+2)=f(2n+2)$, hence $\tilde{f}$ is an extension of $f$.
We have shown that every function in $B$ has an extension in $\mathcal{W}$.
It now remains to show that we can find a peaceful extension inside $\mathcal{W}$.
Consider Eq. \eqref{eq:class_W_3} with $x\in[2n,2n+1]$. We have

\begin{myequation}
f(x) &= f(2n)+[f(2n+1)-f(2n)](p(x+1) - n)\\
& = f(2n)+[f(2n+1)-f(2n)]\cdot\epsilon, \quad for\; \epsilon\in[0,1]\\
& = \epsilon f(2n+1) + (1-\epsilon) f(2n)
\end{myequation}

This latter equation shows that $f(x)\in[f(2n),f(2n+1)]$. Similarly, from Eq. \eqref{eq:class_W_4}
for $x\in[2n+1,2n+2]$ we have

\begin{myequation}
f(x) &= f(2n+1)+[f(2n+2)-f(2n+1)](p(x)-n)\\
& = f(2n+1)+[f(2n+2)-f(2n+1)]\cdot\delta, \quad for\; \delta\in[0,1]\\
& = \delta f(2n+2) + (1-\delta) f(2n+1)
\end{myequation}

This latter equation shows that $f(x)\in[f(2n+1),f(2n+2)]$. The latter two equations then
imply that every function generated by the safe integration operator is peaceful.
Now consider an arbitrary function $f\in\mathcal{W}$.
Define the following functions in $\mathcal{W}$:

\begin{myequation}
&\hat{g}()=f(0)\\
&\hat{h}_0(x;y)=f(2x)\\
&\hat{h}_1(x;y)=f(2x+1)
\end{myequation}

Now define a function $\hat{f}\in\mathcal{W}$ by safe integration using the functions $\hat{g},\hat{h}_1$, and $\hat{h}_2$.
It can be easily seen that $\hat{f}(n)=f(n)$ for every $n\in\N$.
In addition $\hat{f}$ is peaceful. This completes the proof of the third part of the proposition.
\end{proof}

The previous proposition indicates that $\mathcal{W}$ is a class of polytime computable real functions that
peacefully approximates polytime computable integer functions. Hence, using Theorem \ref{th:lip}
the following result is obtained.

\begin{theorem}\label{th:w1}
A Lipschitz function $f:\R \to \R$ is polytime computable iff it is $\mathcal{W}$-definable.
\end{theorem}

Additionally, the previous proposition implies that any function in $\mathcal{W}$ has a polytime computable integer approximation (since $\mathcal{W}$ preserves the integers), hence using Corollary \ref{machinchose}, we can get the following result.

\begin{theorem}\label{th:w2}
Let $f:\R \to \R$ be some $n^k$-smooth function for some $k$.
Then $f$ is polytime computable iff it is $n^k$-$\mathcal{W}$-definable.
\end{theorem}

Notice that $\mathcal{C}$-definability of a function can be seen as a schema that builds a function $f$ from a function $\tilde{g}\in\mathcal{C}$ (see Definition \ref{dfn:approximation_1}).
Let $\mathit{Def}[\mathcal{C}]$ stand for $\mathcal{C}$-definability.
That is, $\mathit{Def}[\mathcal{C}]$ is the class
of functions $f$ such that $f$ is $\mathcal{C}$-definable.
Similarly, given a function $T\colon\N\to\N$, let $T\text{-}\textit{Def}[\mathcal{C}]$
denote $T\text{-}\mathcal{C}$-definability, that
is, $T\text{-}\textit{Def}[\mathcal{C}]$ is the class of functions $f$ such that $f$ is $T\text{-}\mathcal{C}$-definable.
Then, the class of polytime computable functions can be algebraically characterized in a machine-independent way as follows.

\begin{corollary}
\label{cor:w3}
A function $f\colon\R\to\R$ is polytime computable  iff
either (1) $f$ is Lipschitz and belongs to
$\mathit{Def}[\mathcal{W}]$, (2) $f$ is locally
poly-Lipschitz and belongs to $\mathit{Def}[\mathcal{W}]$, or
(3) $f$ is $n^k$-smooth and belongs to $n^k\text{-}\textit{Def}[\mathcal{W}]$ for some $k\in\N$.
\end{corollary}





\bibliographystyle{elsarticle-num}
\bibliography{temp}







\end{document}